\documentclass[12pt,twoside]{article}
\usepackage{amsmath,amsfonts,amsthm,amssymb,bm,bbold,hyperref,times,enumerate}

\newtheorem{thm}{Theorem}[section]
\newtheorem*{thm*}{Theorem}
\newtheorem{lemma}[thm]{Lemma}
\newtheorem{prop}[thm]{Proposition}
\newtheorem{cor}{Corollary}[thm]
\newtheorem*{rmk*}{Remark}

\def\dist{\mathrm{dist}}
\def\m{{\bf m}}
\def\n{{\bf n}}
\def\w{{\bf w}}
\def\lam0{{\lambda+i0}}

\numberwithin{equation}{section}
\begin{document}

\title{Localisation for non-monotone Schr\"odinger operators}

\author{Alexander Elgart\textsuperscript{1}, Mira Shamis\textsuperscript{2}, and Sasha Sodin\textsuperscript{3}}
\footnotetext[1]{Department
of Mathematics, Virginia Tech., Blacksburg, VA, 24061, USA.  E-mail: aelgart@vt.edu.  Supported in part by NSF
grant DMS--0907165.}
\footnotetext[2]{Institute for Advanced Study,
Einstein Dr., Princeton, NJ 08540, USA, and Princeton University,
Princeton, NJ 08544, USA. E-mail: mshamis@princeton.edu. Supported by NSF
grants DMS-0635607 and PHY-1104596.}
\footnotetext[3]{Institute for
Advanced Study, Einstein Dr., Princeton, NJ 08540, USA. E-mail:
sodinale@ias.edu. Supported by NSF under agreement DMS-0635607.}
\maketitle

\begin{abstract}
We study localisation effects of strong disorder on the spectral and
dynamical properties of (matrix and scalar) Schr\"odinger operators with
non-monotone  random potentials, on the $d$-dimensional lattice.
Our results include dynamical localisation, i.e.\ exponentially
decaying bounds on the transition amplitude in the mean.  They are
derived through the study of fractional moments of the resolvent,
which are finite due to resonance-diffusing effects of the
disorder. One of the byproducts of the analysis is a nearly optimal
Wegner estimate. A particular example of the class of systems
covered by our results is the discrete alloy-type Anderson model.

\end{abstract}

\section{Introduction}
\subsection{Random Schr\"odinger operators}
The prototypical model for the study of localisation properties of quantum states of
single electrons in disordered solids is the Anderson Hamiltonian $H_A$, which acts on
$\ell^2(\mathbb{Z}^d)$ by
\[ (H_A\psi)({\bf n}) = v({\bf n}) \psi({\bf n})
    + g^{-1} \sum_{\text{${\bf m}$ adjacent to ${\bf n}$}} \psi({\bf n})~, \]
where the entries $v({\bf n})$ of the potential  are random and independent.

The basic phenomenon, named Anderson localisation after the physicist P. W. Anderson,
is that disorder can cause localisation of electron states, which manifests itself in
time evolution (non-spreading of wave packets), (vanishing of) conductivity in
response to electric field,
Hall currents in the presence of both magnetic and electric field, and statistics
of the spacing between nearby energy levels. The first property
implies spectral localisation, i.e.\ the spectral measure of $H_A$ is almost
surely pure point, and almost sure exponential decay of eigenfunctions.

These properties are known to hold for  $H_A$ in each of the
following cases: 1) high disorder (the coupling constant $g$ is
large), 2) extreme energies, 3) weak disorder away from the spectrum
of the unperturbed operator, and 4) one dimension, $d = 1$.

Historically, the first proof of spectral localisation was given by
Goldsheid, Molchanov and Pastur \cite{GMP}, for a one-dimensional continuous random
Schr\"odinger operator.

In higher dimension, the absence of diffusion was first  established in 1983 by Fr\"ohlich
and Spencer \cite{FS} using multi-scale analysis. Their approach has led to a multitude of
results on localisation for a wide range of problems. The reader is referred to the
monograph of Stollmann \cite{Stollmann} or the recent lecture notes
of W.~Kirsch \cite{Kirsch} for a review of the history of the
subject and a gentle introduction to the multi-scale analysis ---
which is not used here.

One of the ingredients of multi-scale analysis is the regularity of the {\em integrated
density of states}, the (distribution function of the) average of the spectral measure over the randomness.

Ten years later Aizenman and Molchanov \cite{AM} introduced an alternative method for the
proof of localisation, known as  the fractional moment method, which has also found numerous
applications. In particular, in \cite{A}, Aizenman introduced the notion of eigenfunction
correlator, which, combined with the fractional moment method, allowed him to give the
first proof of dynamical localisation. We refer to the lecture notes of Stolz \cite{Stolz2010}
and Aizenman and Warzel \cite{AW} for a survey of subsequent developments.

In the fractional moment method, an a priori estimate on the diagonal elements of the resolvent
$(H_A - \lambda)^{-1}$ plays a key r\^ole in the underlying analysis.

In many situations, regularity of the integrated density of states follows from
the regularity of the distribution of the potential. This was first
proved by Wegner \cite{W}, therefore regularity estimates on the density
of states are called Wegner estimates. An essential ingredient in his argument
is the monotone dependence of the spectrum of $H_A$ on the random variables
$v({\bf n})$. A modification of this argument was applied by Aizenman and Molchanov
to give an a priori bound on the average of $|(H_A-\lambda-i0)^{-1}(x,x)|^s$.

The monotone dependence of the spectrum of the random variables is also used
in the fractional moment's proof of dynamical localisation (via a variant of spectral averaging).

Recently, several challenging problems (listed below) arose in different contexts,
in which the dependence of the spectrum on the random variables is not monotone.
In this paper, we develop a strategy to prove localisation which
is applicable to some of these models.

%If the dependence of the operator on the random variables is not monotone, already
%the identification of the spectral minimum is a highly non-trivial question.
%For more intricate properties, like the regularity of the density of states or the
%analysis of the dynamical behavior, the difference between monotone and non-monotone
%models is even more striking.

%Recently, several problems where Wegner-type arguments do not apply have
%attracted some attention. The challenge is to establish a Wegner estimate
%which is sufficient to obtain localisation.

\vspace{.5cm}\noindent
One close relative of the original Anderson Hamiltonian $H_A$ is the random {\em alloy-type} model, in which the potential
$V({\bf n})$ at a site $\n \in \mathbb{Z}^d$ is obtained from independent random variables $v({\bf m})$
via the formula
\begin{equation}\label{eq:alloy} V({\bf n}) = \sum_{{\bf k}\in\mathbf{\Gamma}} a_{\bf n-k} v({\bf k})~,\end{equation}
where the index ${\bf k}$ takes values in some sub-lattice $\mathbf{\Gamma}$ of $\mathbb{Z}^d$. If all the coefficients
$a_{\bf k}$ have the same sign, the system is  monotone, since the dependence of the spectrum on $V$ is monotone.
Localisation in such systems is well understood by now, even in the continuum setting. The
existing technology is however not well suited to the non-monotone case,
i.e.\ when $a_{\bf k}$  are not all of the same sign. Mathematically,
the problem becomes especially acute when $\sum a_{\bf k} = 0$.

There is no physically compelling reason for a random tight binding alloy model to be
monotone, and  the natural question is whether Anderson
localisation still holds if one breaks the monotony.

\vspace{.5cm}\noindent
Non-monotone models also naturally appear in the class of {\it block operators}.
In one such model, introduced by Fr\"ohlich, and studied by Bourgain in \cite{B},
the matrix-valued potential is given by $V({\bf n}) = U({\bf n})^* A U({\bf n})$, where $A$ is a
fixed self-adjoint $r \times r$ matrix, and $U({\bf n})$ are independently chosen
according to the Haar measure on $SU(r)$. Bourgain proved a volume-dependent
Wegner estimate and Anderson localisation near the edges of the spectrum
using methods from complex analysis.

Another class of non-monotone random block operators, associated with the 
Bogoliubov -- de Gennes symmetry classes, was studied by Kirsch, Metzger, and 
M\"uller in \cite{KMM} and by  Gebert and M\"uller in \cite{GM}. The former paper establishes the Wegner estimate for such random operators in a neighborhood of the inner band edges, by recovering some form of the monotonicity.
The latter paper uses the bootstrap multi-scale analysis to  prove dynamical localization in the same regime.

\vspace{2mm}\noindent
The original motivation for this work was to study a problem suggested by Tom Spencer, in which the
matrix-valued potential is of the form
\[ V({\bf n}) = \left( \begin{array}{cc} v({\bf n}) & a \\ a & - v({\bf n}) \end{array} \right)~, \]
where the variables $v({\bf n})$ are independent and identically distributed. If the distribution of
$v({\bf n})$ has bounded density, the eigenvalue distribution of a single $V({\bf n})$ is $1/2$-H\"older;
this is optimal, since the density of the eigenvalue distribution diverges as $|\lambda \mp a|^{-1/2}$
at the energies $\pm a$.  Spencer conjectured that the integrated density of states for the full Hamiltonian
is also at least $1/2$-H\"older.

\vspace{.5cm}\noindent
Electromagnetic Schr\"odinger operators with random magnetic field, the random displacement model, random block operators and
Laplace-Beltrami operators with random metrics are other examples
of systems with non-monotone parameter dependence which were intensively studied recently.

We refer to the paper of Elgart, Kr\"uger, Tautenhahn, and Veseli\'c \cite{EKTH} for a survey of
recent results on the  sign indefinite alloy-type models with ${\mathbf \Gamma}={\mathbb Z}^d$ and
the bibliography pertaining to some of the models mentioned in the previous paragraph. 

\vspace{.5cm} {\it Summary of results.} For the last few years there has been a continuous 
effort to bring the understanding of models with non-monotone dependence on the 
randomness to a same level as the one for
monotone models. In this paper, we present a method to prove Anderson localisation
and a Wegner estimate for several non-monotone models,  achieving this goal. Theorem~\ref{thm} and its corollaries pertain
to a class of models with matrix-valued potentials; when applied to Spencer's model, it shows that,
under some assumptions on the distribution of $v({\bf n})$, the integrated density of states is
$(1/2 - \epsilon)$-H\"older for any $\epsilon > 0$, at large coupling $g$. Unfortunately, Theorem~\ref{thm}
does not directly apply to Fr\"ohlich's model.

Theorem~\ref{thm:alloy} and its corollaries establish Anderson localisation and a Wegner estimate
for the alloy-type model (\ref{eq:alloy}), in the case that $a$ is finitely supported.

Our argument can be viewed as a further augmentation  of the
fractional moment method of Aizenman--Molchanov \cite{AM}.  In
particular, Proposition~\ref{prop1} is a modification of
\cite[(2.25)]{AM}, whereas Proposition~\ref{prop3} is a version of
the decoupling estimates \cite[Lemmata~2.3,3.1]{AM}. The 
innovation of this work is that we do not rely on an a priori
estimate on the moments of diagonal resolvent elements; instead, we
prove such an estimate in parallel with localisation. 
We also propose an argument which allows to deduce dynamical localisation
directly from the resolvent estimates, and which works in the non-monotone
setting as well as in the monotone one.\footnote{After this work was completed, we have learned from Victor Chulaevsky that a similar argument was used earlier by Germinet and Klein in the context of the multi-scale analysis, see \cite[Theorem~1, Corollary~1]{GK}.}

\vspace{.5cm}
{\it Relation to some past and present works.}
For the one dimensional continuum alloy-type random models the proof of the complete Anderson localization
was first given by Stolz \cite{Stolz2000}.

Outside the spectrum of the unperturbed operator (corresponding to the random potential being switched off)
one can obtain Lipschitz regularity of the integrated density of states by reducing the problem to the monotone case.
The optimal Wegner estimate in this case was established by Combes, Hislop, and Klopp \cite{CHK} (in the continuum, but their
argument is equally applicable in the discrete setting). This input can be used to prove Anderson localisation in the
regimes of  extreme energies and weak disorder away from the spectrum of the unperturbed operator (for the latter regime
for the continuum models this result goes back to Klopp \cite{Klopp}).

Our work covers the remaining perturbative setting, namely the high disorder regime, where we prove complete localisation.

Recently, Bourgain (private communication) devised a different approach that allows to prove $s$-regularity
of the density of states for a wide class of non-monotone
models which includes Fr\"ohlich's model, as well as some
of the models we consider in this note.

\subsection{Notation and statement of results}

Let $\mathcal{G} = (\mathcal{V}, \mathcal{E})$ be a graph with degree at most $\kappa$;
the set of vertices (sites) $\mathcal{V}$ may be either finite or countable. The main example
is the lattice $\mathcal{G} = \mathbb{Z}^d$ (where $\kappa = 2d$), however, the greater
generality does not require additional effort here. For $x, y \in \mathcal{V}$, denote
by $\mathrm{dist}(x, y)$ the length of the shortest path connecting $x$ to $y$; when
$\mathcal{G} = \mathbb{Z}^d$,
\[ \mathrm{dist}(x, y) = \|x-y\|_1~.\]

Let
\[ v : \Omega \times \mathcal{V} \longrightarrow \mathbb{R} \]
be a collection of independent, identically distributed random variables, where
$(\Omega, \mathbb{P})$ is a probability space, and we assume that the distribution
$\mu$ of every $v(x)$
\begin{enumerate}
\item[{\bf \underline{A1}}] is $\alpha$-regular for some $\alpha > 0$, meaning that
$\mu[t - \epsilon, t + \epsilon] \leq C_{\bf A1} \epsilon^\alpha$ for any
$\epsilon > 0$ and $t \in \mathbb{R}$;
\item[{\bf \underline{A2}}] has a finite $q$-moment for some $q>0$, meaning that
$\int |x|^q d\mu(x) \leq C_{\bf A2}$.
\end{enumerate}
For example, the Gaussian distribution and the uniform distribution on a finite
interval satisfy {\bf A1} with $\alpha = 1$ and {\bf A2} with any $q > 0$.

We shall denote the expectation by $\langle \cdot \rangle$ and the expectation
over the distribution of one $v(x)$ by $\langle \cdot \rangle_{v(x)}$.

In the electron gas approximation the system of electrons in a crystal is modeled
by a gas of Fermions moving on a lattice. The excitations of the system are described
by an effective one-body Hamiltonian $H$, which consists of a short-range hopping term
and a local (single site) potential. Each site $x$ of the lattice will be assumed to
have $k$ internal  degrees of freedom.

\vspace{.5cm}
\noindent
{\it Single site (matrix) potential}:
For any $x \in \mathcal{V}$, define a Hermitian matrix
\[ V(x) = v(x) A(x) + B(x)~,\]
where the Hermitian $k \times k$ matrices $A(x)$ and $B(x)$ satisfy
\begin{enumerate}
\item[{\bf \underline{B1}}] $\|A(x)\|,\|A(x)^{-1}\| \le C_{\bf B1}$;
\item[{\bf \underline{B2}}] $\|B(x)\|\le C_{\bf B2}$.
\end{enumerate}

\vspace{.5cm}
\noindent
{\it Hopping}: For every ordered pair $(x,y)\in \mathcal{V}\times \mathcal{V}$ of adjacent sites
(i.e.\ $(x, y) \in \mathcal{E}$) we introduce a $k \times k$ matrix ({\em kernel}) $K(x, y)$ so
that
\begin{enumerate}
\item[{\bf \underline{B3}}] $K(y, x)=K(x, y)^*$ and $\|K(x, y)\|\leq C_{\bf B3}$.
\end{enumerate}

\vspace{.5cm}
We are now in position to introduce our one-particle Hamiltonian.
Namely, let $H$ be a random operator acting on $\ell^2(\mathcal{V}) \otimes \mathbb{C}^k$
(the space of square-summable functions $\psi: \mathcal{V} \to \mathbb{C}^k$)
\begin{equation}\label{eq:defH}
(H \psi)(x) = V(x) \psi(x) + g^{-1} \sum_{y \sim x} K(x,y)\psi(y)~,
\end{equation}
where $g>0$ is a coupling constant, and the sum is over all $y \in \mathcal{V}$
such that $(x, y) \in \mathcal{E}$.

Let $G_\lambda = (H - \lambda)^{-1}$ be the resolvent of $H$, $\lambda \notin \mathbb{R}$. It is known
that the limit $G_{\lam0} = \lim_{\epsilon \to +0} G_{\lambda + i \epsilon}$ exists for almost every $\lambda \in \mathbb{R}$. In the following,
$\langle \| G_{\lambda + i0}(x, y) \|^s \rangle$ can a priori be formally interpreted as
\[ \lim_{\epsilon \to +0} \langle \| G_{\lambda + i\epsilon}(x, y) \|^s \rangle~;\]
a posteriori, $G_{\lam0}$ is finite almost surely, and
\[ \lim_{\epsilon \to +0} \langle \| G_{\lambda + i\epsilon}(x, y) \|^s \rangle
    = \langle \| \lim_{\epsilon \to +0} G_{\lambda + i\epsilon}(x, y) \|^s \rangle~.\]

\begin{thm}\label{thm} Let $0 < s \leq \frac{\alpha q}{2k\alpha + kq}$.  There exists
$C>0$ that may depend on $\alpha$, $q$, $C_{\bf A1}$--$C_{\bf B3}$
and $s$ such that for any $\lambda \in \mathbb{R}$ and any $g \geq C \kappa^{1/s} / (1 + |\lambda|)$
\[ \langle \| G_{\lambda + i0}(x, y) \|^s \rangle
    \leq \frac{C}{(1+|\lambda|)^s} \left( \frac{C\kappa}{g^s (1 + |\lambda|)^s} \right)^{\dist(x,y)}~.  \]
\end{thm}

\vspace{1mm}\noindent
Let us state some corollaries for the homogeneous setting; for simplicity, assume that
$\mathcal{G}=\mathbb{Z}^d$ (we denote the vertices of $\mathbb{Z}^d$ by $\m,\n,\cdots$)
We also assume that
\begin{enumerate}
\item[{\bf \underline{C}}] \quad $A(\m)\equiv A$, $B(\m)\equiv B$, $K(\m,\n)\equiv K(\m-\n)$.
\end{enumerate}
The density of states $\rho$ is defined as the average of the spectral measure corresponding to $H$:
\[ \int f(\lambda) d\rho(\lambda) = \frac{1}{k} \textrm{tr} \, \langle f(H)({\bf n}, {\bf n}) \rangle~, \]
where $\mathrm{tr}$ stands for the trace. The integrated density of states is the distribution
function $\lambda \mapsto \rho(-\infty, \lambda]$ of $\rho$.

The assumption {\bf C} guarantees that these definitions do not depend on the choice of the
vertex ${\bf n} \in \mathbb{Z}^d$.

Theorem~\ref{thm} implies the following Wegner-type estimate:

\begin{cor}\label{cor:weg}
Assume {\bf C}. If $g \geq C d^{1/s}/(1+|\lambda|)$, then the integrated density of states is locally
$s$-H\"older at $\lambda$ for
\[ s = \frac{\alpha q}{2k\alpha + kq} = \frac{\alpha}{k \left(1 + \frac{2\alpha}{q}\right)}~,\]
uniformly in $g \to \infty$:
\[ \rho[\lambda - \epsilon, \lambda + \epsilon] \leq C (1 + |\lambda|)^{-s} \epsilon^s~. \]
\end{cor}

\noindent In particular, for any distribution with bounded density and finite moments the integrated density of states
is $1/(k+\epsilon)$-H\"older for any $\epsilon>0$.

Next, Theorem~\ref{thm} implies dynamical and spectral Anderson localisation:

\begin{cor}\label{cor:dyn} Assume {\bf C}. Let $I$ be a finite interval of energies, and let
\[ g \geq \frac{C d^{1/s}}{1 + \min_{\lambda \in I} |\lambda|}~. \]
Then, for any $\m \neq \n \in \mathbb{Z}^d$,
\begin{equation}\label{eq:dynam}
\langle \sup_{t \geq 0} \left|  e^{i t H_I}(\m, \n) \right|  \rangle
    \leq C \dist(\m,\n)^{2d} \left( \frac{C d}{g^s (1+|\lambda|)^s}\right)^\frac{s \, \dist(\m,\n)}{8}~,
\end{equation}
where $H_I = P_I H P_I$, $P_I$ is the spectral projector corresponding to $I$. Therefore the spectrum of $H$ in $I$ is
almost surely pure point.
\end{cor}

The first part of the last corollary follows from Theorem~\ref{thm}, (\ref{eq:obv}), and Theorem~\ref{thm:dyn}.
The ``therefore'' part follows from the summability of the right-hand side of (\ref{eq:dynam})
via the Kunz--Souillard theorem \cite[Theorem~9.21]{CFKS}.

\subsection{Extensions}

\paragraph{Alloy-type models.} Consider the operator $H$ with potential (\ref{eq:alloy}) acting on
$\ell^2(\mathbb{Z}^d)$. Let $\mathcal{B}_\n$ be the set of $v({\bf m})$ for which $a_{\n-{\bf m}} \neq 0$.
We will impose the following  assumptions on the random potential:
\begin{enumerate}
\item the set $\mathcal{B}_\n$ is  non empty for all $\n$;
\item the cardinality $k = \# \{ \m \, | \, a_{\bf \m} \neq 0 \} < \infty$;
\item the distribution of $v(\m)$ satisfies {\bf A1} and {\bf A2}.
\end{enumerate}

\begin{thm}\label{thm:alloy} Let $0 < s < \frac{\alpha q}{2k\alpha + kq}$. There exists $C>0$
such that for any $\lambda \in \mathbb{R}$ and any $g \geq C d^{1/s}/(1+|\lambda|)$
\[ \langle |G_\lam0(\m, \n)|^s \rangle
    \leq \frac{C}{(1+|\lambda|)^s} \left( \frac{Cd}{g^s (1+|\lambda|)^s }\right)^{\dist(\m,\n)}~. \]
\end{thm}

Similarly to Corollaries~\ref{cor:weg},\ref{cor:dyn}, one can deduce a Wegner estimate and
Anderson localisation.
\paragraph{Relaxing the covering condition.} The assumption  $\bf B1$ is usually referred to as a covering condition.
In our analysis, it enters in the proof of Lemma \ref{lem1}. In particular, all our results are still valid
(albeit with the less sharp bound  on the underlying localisation length) if one replaces  $\bf B1$ with
\begin{enumerate}
\item[{\bf \underline{B1$'$}}] $\langle \| (V(y)-\lambda)^{-1} \|^{s}\rangle_{v(y)} \leq C g^{\alpha}$, with $\alpha< s$.
\end{enumerate}
For a fixed non zero matrix $A(y)$ and a generic matrix $B(y)$ the estimate $\bf B1'$ holds true for $g$ large enough.
For instance $\bf B1'$ is applicable (with $\alpha=0$) for
\[ A(y) = \left( \begin{array}{ccc} 1 & 0 & 0 \\
0 & 0 &0 \\ 0 & 0 &-1 \end{array} \right)~;
\quad  B(y) = \left( \begin{array}{ccc} 0 & 1 & 0 \\ 1 & 0 &2 \\ 0 & 2 &0 \end{array} \right)~.\]

\vspace{5mm}
\paragraph{Acknowledgments.} We are grateful to Tom Spencer for suggesting the problem and for helpful
discussions, to Michael Aizenman for comments and suggestions, in particular, for suggesting
to use eigenfunction correlators in the proof of dynamical localisation, to G\"unter Stolz
for remarks on a preliminary version of this paper, and to Victor Chulaevsky for bringing
the work of Germinet and Klein \cite{GK} to our attention.

\section{Proof of theorems}

In the proof of Theorem~\ref{thm}, we assume that the graph $\mathcal{G}$ is finite
and $\lambda \notin \mathbb{R}$. The estimates will be uniform in $\# \mathcal{V} \to \infty$
($\#$ denotes cardinality) and $\mathrm{Im}\, \lambda \to 0$, therefore the statement for infinite
graphs and real $\lambda$ can be deduced as follows. First, an infinite graph can be approximated
by its finite pieces; the matrix elements of the resolvent corresponding to the finite pieces
converge to the matrix elements of the resolvent corresponding to the infinite graph, yielding
the same estimate for $\lambda \notin \mathbb{R}$. Then one can let $\mathrm{Im} \, \lambda$ go
to zero.

The proof of Lemma~\ref{lem1}  below will be postponed until Section~\ref{s:decoupling}.

\begin{prop}\label{prop1} For any $s \leq \frac{\alpha q}{2k\alpha + kq}$ there exists $C>0$ (depending
on $s$ and the constants in the assumptions) such that
for any $\lambda \notin \mathbb{R}$
\[ \langle \|G_\lambda(x, y)\|^s \rangle
    \leq \frac{C}{2(1+|\lambda|)^s} \left\{ g^{-s} \sum_{z \sim y} \langle \|G_\lambda(x, z)\|^s \rangle
    + \delta_{xy} \right\}~,\]
where
\[ \delta_{xy} = \begin{cases} 1, &x = y \\ 0, &x \neq y \end{cases} \]
is the Kronecker $\delta$.
\end{prop}

\begin{proof} By definition of $G_\lambda$,
\[ G_\lambda(x, y) (V(y) - \lambda) = - g^{-1} \sum_{z \sim y} G_\lambda(x, z) K(z,y) + \delta_{xy}~.\]
Therefore
\[ \langle \|G_\lambda(x, y) (V(y) - \lambda)\|^s \rangle
    \leq C_{\bf B3}^s g^{-s} \sum_{z \sim y} \langle \|G_\lambda(x, z)\|^s \rangle + \delta_{xy}~. \]

\begin{lemma}\label{lem1} For $s \leq \frac{\alpha q}{2k\alpha + kq}$, there exists  $\hat C$ (depending on $s$ and the constants in the assumptions) such that
\[ \langle \| G_\lambda(x, y) (V(y) - \lambda) \|^s \rangle
    \geq {\hat C}^{-1} \langle \|G_\lambda(x, y)\|^s \rangle \, (1+|\lambda|)^s~. \]
\end{lemma}

\noindent The proposition follows.
\end{proof}

\begin{cor}\label{cor:max}  For any $s \leq \frac{\alpha q}{2k\alpha + kq}$, we have
\[ \max_y \langle \|G_\lambda(x, y)\|^s \rangle = \langle \|G_\lambda(x, x)\|^s \rangle~,\]
provided $g^s \geq C\kappa/(1+|\lambda|)^s$.
\end{cor}

\begin{proof}
Suppose the maximum $M$ is attained at $y \neq x$. Then
\[\begin{split}
M &= \langle \|G_\lambda(x, y)\|^s \rangle
    \leq \frac{C}{2 g^s (1+|\lambda|)^s} \sum_{z \sim y} \langle \|G_\lambda(x, z)\|^s \rangle \\
  &\leq \frac{C \kappa M}{2 g^s (1+|\lambda|)^s} \leq \frac{CM}{2C}=\frac{M}{2}~,
\end{split}\]
a contradiction.
\end{proof}

\begin{cor}\label{cor2} For any $s \leq \frac{\alpha q}{2k\alpha + kq}$ and $g^s \geq C\kappa/(1+|\lambda|)^s$
\[ \langle \|G_\lambda(x, x)\|^s \rangle \leq \frac{C}{(1+|\lambda|)^s}~.\]
\end{cor}

\begin{proof}
By Proposition~\ref{prop1} with $y = x$ and Corollary~\ref{cor:max},
\[\begin{split} \langle \|G_\lambda(x, x)\|^s \rangle
    &\leq \frac{C}{2(1+|\lambda|)^s} \left\{ g^{-s} \kappa \langle \|G_\lambda(x, x)\|^s \rangle + 1 \right\}\\
    &\leq \frac{1}{2} \langle \|G_\lambda(x, x)\|^s \rangle + \frac{C}{2 (1+|\lambda|)^s}~,\end{split}\]
therefore
\[ \langle \|G_\lambda(x, x)\|^s \rangle \leq \frac{C}{(1+|\lambda|)^s}~.\]
\end{proof}

\begin{proof}[Proof of Theorem~\ref{thm}] For $x = y$ the inequality follows from Corollary~\ref{cor2}. For
$x \neq y$ apply Proposition~\ref{prop1} $\text{dist}(x,y)$ times, and then use Corollary~\ref{cor:max} and
Corollary~\ref{cor2} to estimate every term.
\end{proof}

\begin{proof}[Proof of Theorem~\ref{thm:alloy}]

The proof follows that of Theorem~\ref{thm}. The main modification
(apart from replacing $\| \cdot \|$ with $| \cdot |$) appears in Lemma~\ref{lem1},
which has to be replaced with

\begin{lemma}\label{lem1'}
For $s \leq \frac{\alpha q}{2k\alpha + kq}$, there exists $\hat{C}$ such that
\[ \langle |G_\lambda(\m, \n)|^s |V(\n) - \lambda|^s  \rangle
    \geq \hat{C}^{-1} \langle |G_\lambda(\m, \n)|^s  \rangle (1+|\lambda|)^s~.\]
\end{lemma}

\noindent The proof is provided at the end of Section~\ref{s:decoupling}.

\end{proof}

\section{Estimates on ratios of polynomials}\label{s:decoupling}

Lemma~\ref{lem1}  will follow from

\begin{prop}\label{prop3} Let  $\mu$ be a probability measure satisfying the assumptions {\bf A1}, {\bf A2}.
Let $a_1, \cdots, a_l, b_1, \cdots, b_m \in \mathbb{C}$, and let $s, r>0$ be such
that $rm < \alpha$ and $q \geq (sl+rm) \frac{\alpha}{\alpha - rm}$. Then
\[ \int \frac{\prod_{j=1}^l |v - a_j|^s}{\prod_{i=1}^m|v-b_i|^r} \, d\mu(v)
    \asymp \frac{\prod_{j=1}^l (1+|a_j|)^s}{\prod_{i=1}^m (1  +|b_i|)^r}~,\]
where the "$\asymp$" sign means that $\text{LHS} \leq C \, \text{RHS} \leq C' \,\text{LHS}$, and the numbers
$C,C'>0$ may depend on $\alpha$, $q$, $C_{\bf A1}$, $C_{\bf A2}$, $l$, $m$, $r$, and $s$, but not on
$a_j$ and $b_i$.
\end{prop}

\begin{proof}[Proof of Lemma~\ref{lem1}]

First let us show that the statement holds for very small $s>0$; then we shall extend it to all
$s \leq \frac{\alpha q}{4\alpha + 2q}$. We shall consider the (slightly more complicated) case $x \neq y$.

For $s$ sufficiently small,
\begin{equation}\label{eq:schw}\begin{split}
&\langle \|G_\lambda(x, y)\|^s \rangle_{v(y)} \\
    &\quad\leq \langle \|G_\lambda(x, y) (V(y) - \lambda) \|^s \| (V(y)-\lambda)^{-1} \|^s \rangle_{v(y)} \\
    &\quad\leq \langle \|G_\lambda(x, y) (V(y) - \lambda)\|^{2s} \rangle_{v(y)}^{1/2}
        \langle \| (V(y)-\lambda)^{-1} \|^{2s}\rangle_{v(y)}^{1/2}~.
\end{split}\end{equation}
By the Schur--Banachiewicz formula for the inverse of a block matrix\footnote{See Henderson and Searle \cite{HS}
for the history of block matrix inversion formulae},
\[ \left( \begin{array}{cc} G_\lambda(x, x) & G_\lambda(x, y) \\ G_\lambda(y, x) & G_\lambda(y, y) \end{array} \right)
    = \left[ \left(\begin{array}{cc} V(x) & \\ & V(y) \end{array}\right) - K_{2k \times 2k} \right]^{-1}~,\]
where $K_{2k\times2k}$ is independent of $v(x), v(y)$. Applying the Schur--Banachiewicz formula once
again, we obtain:
\[ G_\lambda(x, y) = L_{k\times k} (V(y) - M_{k\times k})^{-1}
    = \frac{L_{k\times k} (V(y) - M_{k\times k})^\text{Adj}}{\det(V(y) - M_{k  \times k})}~,\]
where $L_{k\times k}$ and $M_{k\times k}$ are independent of $v(y)$, and $\text{Adj}$ denotes
the adjugate (= cofactor) matrix.
\[ G_\lambda(x, y) (V(y) - \lambda)
    = \frac{L_{k\times k} (V(y) - M_{k \times k})^\text{Adj} (V(y)-\lambda)}{\det(V(y) - M_{k \times k})}~.\]
Therefore every entry of $ G_\lambda(x, y) (V(y) - \lambda) $ is a ratio of two polynomials $Q_1, Q_2$ of degree $\le k$
with respect to the variable $v(y)$.
For sufficiently small $s>0$, Proposition~\ref{prop3} implies that for any such pair  $Q_1, Q_2$
\[ \left\{ \int \frac{|Q_1(v)|^{2s}}{|Q_2(v)|^{2s}} d\mu(v) \right\}^{1/2}
    \leq \tilde C \int \frac{|Q_1(v)|^{s}}{|Q_2(v)|^{s}} d\mu(v)~. \]
Hence
\[ \langle \|G_\lambda(x, y) (V(y) - \lambda)\|^{2s} \rangle_{v(y)}^{1/2}
    \leq k\tilde C \langle \|G_\lambda(x, y) (V(y) - \lambda)\|^{s} \rangle_{v(y)}~. \]
Proposition~\ref{prop3} also implies that for sufficiently small $s$
\begin{equation}\label{eq:Vinv} \langle \| (V(y)-\lambda)^{-1} \|^{2s}\rangle_{v(y)}^{1/2} \leq 2k(1 + |\lambda|)^{-s}~.\end{equation}
Indeed, using Proposition~\ref{prop3} we first observe that for sufficiently small $s$
\[ \langle \| (v(y) A(y)+B(y)+i)(V(y)-\lambda)^{-1} \|^{2s}\rangle_{v(y)}^{1/2} \leq 1.1k~,\]
 Using now the resolvent identity
 \[(V(y)-\lambda)^{-1}=-(i+\lambda)^{-1}+(i+\lambda)^{-1}(v(y) A(y)+B(y)+i)(V(y)-\lambda)^{-1}~,\]
 we establish \eqref{eq:Vinv}.

Returning to (\ref{eq:schw}), we obtain:
\begin{equation}\label{eq:comp}
\langle \|G_\lambda(x, y)\|^s \rangle_{v(y)}
    \leq \hat C \langle \|G_\lambda(x, y) (V(y) - \lambda)\|^{s} \rangle_{v(y)}
    (1+|\lambda|)^{-s}~.
\end{equation}
To extend this inequality to all $s \leq \frac{\alpha q}{4\alpha + 2q}$, we apply Proposition~\ref{prop3}
once again. Every entry in $G_\lambda(x, y)$ and $G_\lambda(x, y) (V(y) - \lambda)$  is a ratio
of two polynomial functions of $v(y)$ whose degree do not exceed $k$. By Proposition~\ref{prop3}, the expressions
\[ \left\{\int \frac{|Q_1(v)|^{s}}{|Q_2(v)|^{s}} d\mu(v) \right\}^{1/s} \]
are comparable as long as $q \geq \frac{2ks \alpha}{\alpha - ks}$, that is, $s \leq \frac{q\alpha}{kq+2k\alpha}$.
Therefore (\ref{eq:comp}) remains valid in this range of $s$. Averaging over $(v(z))_{z \neq y}$,
we obtain
\[ \langle \|G_\lambda(x, y)\|^s \rangle \leq \hat C \langle \|G_\lambda(x, y) (V(y) - \lambda)\|^{s} \rangle (
1+|\lambda|)^{-s}~.\]
\end{proof}

\begin{proof}[Proof of Proposition~\ref{prop3}] \hfill

\vspace{2mm}\noindent{\bf Lower bound.} Choose $R>0$ so that $\mu[-R, R] \geq 1/2$ (for example, take
$R = \max(1, 2 C_{\bf A2})^{1/q}$.) Then

\[\begin{split}
\int \frac{\prod_{j=1}^l |v - a_j|^s}{\prod_{k=1}^m|v-b_k|^r} \, d\mu(v)
    &\geq \int_{-R}^R \\
    &\geq C_1^{-1} \frac{\prod_{|a_j| \geq 2R} (1+|a_j|)^s}{\prod_{k=1}^m (1  +|b_k|)^r}
        \int_{-R}^R \prod_{|a_j| < 2R} |v - a_j|^s d\mu(v)~.
\end{split}\]
Now, for any $0 < t < 1$, the set $\{ \prod |v - a_j| \leq t\}$ can be covered by $l$ intervals
of length $C t^{1/l}$. Therefore, when $t>0$ is sufficiently small,
\[\begin{split}
\mu \left\{ \prod_{|a_j|<2R} |v - a_j| \leq t \right\}  \leq C_2 t^{\alpha/l} \leq 1/4~.
\end{split}\]
Then
\[ \int_{-R}^R \prod_{|a_j| < 2R} |v - a_j|^s d\mu(v) \geq t^s/4 \geq C_3^{-1} \geq C_4^{-1} \prod_{|a_j|<2R} (1+|a_j|)^s~. \]

\vspace{2mm}\noindent{\bf Upper bound.} Let us start with several reductions. First, it is sufficient
to consider the case $a_1 = \cdots = a_l = a$, $b_1 = \cdots = b_m = b$. This follows from the Cauchy--Schwarz inequality
\[ \int \frac{\prod_{j=1}^l |v - a_j|^s}{\prod_{k=1}^m|v-b_k|^r} \, d\mu(v)
    \leq \prod_{j=1}^l \prod_{k=1}^m \left\{ \int \frac{|v - a_j|^{sl}}{|v-b_k|^{rm}} \, d\mu(v)\right\}^{\frac{1}{lm}}~.  \]
Second,
\[ \int \frac{|v - a|^{sl}}{|v-b|^{rm}} \, d\mu(v)
    \leq C \left\{ \int \frac{|v|^{sl} \, d\mu(v)}{|v-b|^{rm}} + |a|^{sl} \int \frac{d\mu(v)}{|v-b|^{rm}}\right\}~,\]
so it is sufficient to consider the case $a = 0$. Third, we can assume that $|b|>1$, since for $|b|\leq1$ the regularity
condition {\bf A1} implies
\[ \int \frac{|v|^{sl} \, d\mu(v)}{|v-b|^{rm}}
    \leq C \left\{ \int \frac{d\mu(v)}{|v-b|^{rm-sl}} + |b|^{sl} \int \frac{d\mu(v)}{|v-b|^{rm}} \right\}
    \leq C_5 \leq \frac{C_6}{(1+|b|)^{rm}}~.  \]
Therefore we need to show that for $|b|>1$
\[ \int \frac{|v|^{sl} \, d\mu(v)}{|v-b|^{rm}} \leq C |b|^{-rm}~.\]

\noindent Let us divide the integral into two parts:
\[ \int = \int_{\big||v|-|b|\big|>|b|/2} + \int_{|b|/2 < |v| < 3|b|/2}~.\]
By {\bf A2}, the first integral is at most
\[ \left(\frac{2}{|b|}\right)^{rm} \int |v|^{sl} d\mu(v) \leq C_6 |b|^{-rm}~. \]
Let us estimate second integral.
\begin{equation}\label{eq:est4}\begin{split} \int_{|b|/2 < |v| < 3|b|/2}
    &\leq \left(\frac{3|b|}{2}\right)^{sl} \int_{|b|/2 < |v| < 3|b|/2} \frac{d\mu(v)}{|v-b|^{rm}} \\
    &\leq C_7 |b|^{sl} \int_0^\infty \mu \left\{ |b|/2 < |v| < 3|b|/2~, \, |v-b| < t^{-\frac{1}{rm}}\right\} \, dt \\
    &= C_7 |b|^{sl} \left\{ \int_0^{b^{-sl-rm}} + \int_{b^{-sl-rm}}^{b^\gamma}+ \int_{b^\gamma}^\infty \right\}~,
\end{split}\end{equation}
where $\gamma > 0$ is a number that we shall choose shortly. The first integral in (\ref{eq:est4}) is
at most $b^{-sl-rm}$. The second integral is at most
\[  b^\gamma \mu \left\{ |v|>|b|/2 \right\} \leq C_8 b^{\gamma - q} \leq C_8 b^{-sl-rm} \]
as long as
\begin{equation}\label{eq:cond.gamma1} \gamma \leq q - sl -r m~.\end{equation}
The third integral is at most
\[ \int_{b^\gamma}^\infty \mu \left\{|v-b| < t^{-\frac{1}{rm}}\right\} \, dt
    \leq C_9 \int_{b^\gamma}^\infty t^{-\frac{\alpha}{rm}}  \, dt
    \leq C_{10} |b|^{-\gamma (\frac{\alpha}{rm} - 1)}  \leq C_{11} |b|^{-sl-rm}\]
as long as
\begin{equation}\label{eq:cond.gamma2} \gamma \geq \frac{sl+rm}{\frac{\alpha}{rm}-1}~.\end{equation}
Since $q \geq (sl+rm) \frac{\alpha}{\alpha - rm}$, we can choose $\gamma$ that satisfies
both (\ref{eq:cond.gamma1}) and (\ref{eq:cond.gamma2}); then we obtain the claimed estimate.
\end{proof}

Now we prove Lemma~\ref{lem1'}.

\begin{proof}
We shall prove that
\[ \langle |G_\lambda(\m, \n)|^s |V(\n) - \lambda|^s  \rangle_{\mathcal{B}_\n}
    \geq \hat{C}^{-1} \langle |G_\lambda(\m, \n)|^s  \rangle_{\mathcal{B}_\n} (1+|\lambda|)^s~. \]
First,
\[ \langle |G_\lambda(\m, \n)|^{s/2}  \rangle_{\mathcal{B}_\n}^2
    \leq \langle |G_\lambda(\m, \n)|^{s} |V(\n) - \lambda|^s \rangle_{\mathcal{B}_\n}
         \langle |V(\n) - \lambda|^{-s} \rangle_{\mathcal{B}_\n}~, \]
and, as above,
\[ \langle |V(\n) - \lambda|^{-s} \rangle_{\mathcal{B}_\n} \leq C (1+|\lambda|)^{-s}~.\]
Therefore it remains to show that
\begin{equation}\label{eq:revhoel}
\langle |G_\lambda(\m, \n)|^{s}  \rangle_{\mathcal{B}_\n}
    \leq C \langle |G_\lambda(\m, \n)|^{s/2}  \rangle_{\mathcal{B}_\n}^2~.
\end{equation}
For simplicity of notation, let $\mathcal{B}_\n = \{ v_1, \cdots, v_J \}$ (here $1 \leq J \leq k$). Cramer's
rule (or the Schur--Banachiewicz formula) shows that, as a function of every $v_j$,
$Q(v_1, \cdots, v_J) = G_\lambda(\m, \n)$ is a ratio of two polynomials
of degree at most $k$. The next multivariate version of Proposition~\ref{prop3} concludes the proof. 
\end{proof}

\begin{prop}
Let $Q$ be a function of $J$ variables $v_1, \cdots, v_J$ which, as a function of every $v_j$, is
a ratio of two polynomials of degree at most $k$. Then, for any probability measure $\mu$ satisfying the assumptions
{\bf A1}, {\bf A2} and for any $s \leq \frac{q\alpha}{k(q+2\alpha)}$,
\begin{multline*} \left\{ \int d\mu(v_1) \cdots d\mu(v_J) |Q(v_1, \cdots, v_J)|^s \right\}^{1/s}\\
	\leq C_{k,\alpha,q,C_{\bf A1},C_{\bf A2}} 
		\left\{ \int d\mu(v_1) \cdots d\mu(v_J) |Q(v_1, \cdots, v_J)|^{s/2} \right\}^{2/s}~.
\end{multline*}
\end{prop}

\begin{proof}
If $J=1$, (\ref{eq:revhoel}) follows from Proposition~\ref{prop3}. Then we proceed
by induction on $J$. By case $J=1$,
\[\begin{split}
&\int d\mu(v_1) \cdots d\mu(v_J) |Q(v_1, \cdots, V_J)|^{s} \\
&\quad\leq C_1 \int d\mu(v_1) \cdots d\mu(v_{J-1}) \left\{ \int d\mu(v_J) |Q(v_1, \cdots, v_J)|^{s/2} \right\}^2\\
&\quad=C_1\int d\mu(v_1) \cdots d\mu(v_{J-1}) \\
    &\qquad\int d\mu(v_J) |Q(v_1, \cdots, v_J)|^{s/2} \int d\mu(v_J') |Q(v_1, \cdots, v_J')|^{s/2} \\
&\quad= C_1\int d\mu(v_J)d\mu(v_J') \int d\mu(v_1) \cdots d\mu(v_{J-1}) \\
    &\qquad|Q(v_1, \cdots, v_J)|^{s/2} |Q(v_1, \cdots, v_J')|^{s/2}
\end{split}\]
By the Cauchy--Schwarz inequality and the induction step, the last expression is at most
\[\begin{split}
&C_1 \int d\mu(v_J)d\mu(v_J')  \\
&\quad    \left\{ \int d\mu(v_1) \cdots d\mu(v_{J-1}) |Q(v_1, \cdots, v_J)|^{s}
        \int d\mu(v_1') \cdots d\mu(v_{J-1}') |Q(v_1', \cdots, v_J')|^{s} \right\}^{1/2} \\
&\leq C_2 \int d\mu(v_J)d\mu(v_J')  \qquad\qquad \\
&\quad    \int d\mu(v_1) \cdots d\mu(v_{J-1}) |Q(v_1, \cdots, v_J)|^{s/2}
    \int d\mu(v_1') \cdots d\mu(v_{J-1}') |Q(v_1', \cdots, v_J')|^{s/2} \\
&\qquad= C_2 \left\{ \int d\mu(v_1) \cdots d\mu(v_J) |Q(v_1, \cdots, v_J)|^{s/2} \right\}^2~.
\end{split}\]
\end{proof}

\section{Dynamical localisation}\label{s:dyn}

The proof of dynamical localisation is based on the notion of eigenfunction correlators, introduced
by Aizenman in \cite{A}.

Let us start with some definitions, which are adjusted from the lecture notes of Aizenman and Warzel \cite{AW}
to our block setting. Let $H$ be an operator acting on $\ell^2(\mathbb{Z}^d)\otimes \mathbb{C}^k$.
For ${\bf m}, {\bf n} \in \mathbb{Z}^d$, the (matrix-valued) spectral measure $\mu_{{\bf mn}}$ is defined by
\[ \int \phi \, d\mu_{\bf mn} = \phi(H)({\bf m,n})~, \quad \phi \in C_0(\mathbb{R})~.\]
The eigenfunction correlator $Q_I({\bf m,n})$ corresponding to a finite interval $I\subset \mathbb{R}$
(on the energy axis) is defined by
\[ Q_I({\bf m,n}) = \sup \left\{ \| \phi(H)({\bf m,n}) \| \, \mid \,
  \mathrm{supp}\, \phi \subset I, \, |\phi| \leq 1 \right\}~. \]
Obviously,
\begin{equation}\label{eq:obv}
\sup_{t \geq 0} |e^{itH_I}(\m,\n)| \leq Q_I(\m, \n)
\end{equation}
for any $t>0$.

The eigenfunction correlators can be also defined for the restriction of $H$ to a finite box $\Lambda$
(we denote this restriction by the superscript $\Lambda$). In this case, it satisfies the following
inequalities (the first one is an equality in the scalar case, cf.\ \cite{AW}):
\begin{lemma}
\[ Q_I^\Lambda(\m,\n)
    \leq \lim_{\epsilon \to +0} \frac{\epsilon}{2}
        \int_I \|G^\Lambda_\lam0(\m, \n)\|^{1-\epsilon} d\lambda
    \leq k~. \]
\end{lemma}

\begin{proof}
For any eigenvalue $\nu$ of $H^\Lambda$, define a $k \times k$ matrix
\[ M_\nu = \sum \psi(\m) \otimes \psi(\n): u \mapsto \sum (\psi(\n) \cdot u) \psi(\m)~, \]
where the sum is over all eigenfunctions $\psi$ of $H^\Lambda$ corresponding to $\nu$. Then
\[ \phi(H^\Lambda)(\m, \n) = \sum_{\nu\in I} \phi(\nu) M_\nu~, \]
whereas
\[ G_\lambda^\Lambda(\m, \n) = \sum_{\nu} \frac{M_\nu}{\nu - \lambda} \]
(where now the sum is over all eigenvalues of $H^\Lambda$.)
Therefore
\[ \| \phi(H^\Lambda)(\m, \n)\| \leq \sum_{\nu \in I} \|M_\nu\|
    = \lim_{\epsilon \to +0} \frac{\epsilon}{2} \int_I \|G_\lambda^\Lambda(\m,\n)\|^{1-\epsilon} d\lambda~. \]
The equality can be proved by representing $I = \uplus I_\nu$ as a disjoint union
of neighbourhoods of $\nu \in I$, and noting than
\[ G_\lambda^\Lambda(\m, \n) = \frac{M_\nu}{\nu - \lambda} + O(1), \quad \lambda \to \nu~.\]
Also,
\[\begin{split}
\sum_{\nu \in I} \|M_\nu\|
    &\leq \sum_\psi \|\psi(\m) \otimes \psi(\n)\|
    = \sum_\psi \|\psi(\m)\| \| \psi(\n) \| \\
    &\leq \left\{ \sum_\psi \|\psi(\m)\|^2 \sum_\psi \|\psi(\n)\|^2 \right\}^{1/2} = k~.
\end{split}\]
\end{proof}

Now assume that $H = H_\omega$ is a random operator of the form (\ref{eq:defH}), where the random
$V({\bf m})$ are independent and identically distributed, $K(\m,\n)$ depends only on $\m-\n$.
We shall prove

\begin{thm}\label{thm:dyn} Let $0 < s < 1$, and suppose for every box $\Lambda$, every $\lambda \in I$,
and every $\m,\n \in \Lambda$
\[ \langle \|G^\Lambda_{\lambda+i0}(\m,\n)\|^s \rangle \leq C \exp(- \gamma \, \dist(\m, \n) ) \]
for some $C, \gamma > 0$, where $G^\Lambda$ is the resolvent of the restriction of $H$ to
$\Lambda$. Then, for every $\m, \n \in \mathbb{Z}^d$,
\[ \langle Q_I(\m, \n) \rangle \leq C' \dist^{2d}(\m,\n) \exp(- \frac{s\gamma}{8} \dist(\m,\n))~.\]
\end{thm}

\begin{rmk*}
A similar statement can be proved for potentials of the form (\ref{eq:alloy}).
\end{rmk*}

\begin{proof} We shall prove the estimate in a large box $\Lambda$ containing $\m,\n$ (uniformly in
the size of $\Lambda$). Let $\Lambda_\m$ and $\Lambda_\n$ be two boxes of radius
$R = \lfloor\dist(\m,\n)/2\rfloor$, centered at $\m$,$\n$, respectively. According to the resolvent identity,
\[ G_\lam0^\Lambda(\m, \n) = g^{-1} \sum_{\w\w' \in \partial \Lambda_\m}
  G_\lam0^{\Lambda_\m}(\m,\w) K(\w,\w') G_\lam0^\Lambda(\w', \n)~,\]
where the sum is over all pairs $\w\w'$ such that $\w \in \Lambda_\m$, $\w' \notin \Lambda_\m$, $\w\sim\w'$. Therefore
\[ \|G_\lam0^\Lambda(\m, \n)\|
  \leq C g^{-1} \max_{\w\w' \in \partial \Lambda_\m} \|G_\lam0^{\Lambda_\m}(\m,\w)\|
            \sum_{\w\w' \in \partial \Lambda_\m} \|G_\lam0^\Lambda(\w', \n)\|~. \]
Now we apply \cite[Prop.~5.1]{ETV} (which holds in the block-operator setting). It shows that,
with probability at least $1 - C' R^{2d}\exp(- \gamma s R/8 )$, one can decompose $I = I_\m \cup I_\n$
so that for every $\w\w' \in \partial \Lambda_\m$ and $\lambda \in I_\m$
\[ \max_{\w\w' \in \partial \Lambda_\m} \|G_\lam0^{\Lambda_\m}(\m,\w)\| \leq C \exp(-\gamma R/8)~, \]
and for every $\w\w' \in \partial \Lambda_\n$ and $\lambda \in I_\n$
\[ \max_{\w\w' \in \partial \Lambda_\n} \|G_\lam0^{\Lambda_\n}(\n,\w)\| \leq C \exp(-\gamma R/8)~.\]
Therefore,
\[ \lim_{\epsilon \to +0} \frac\epsilon2 \int_{I_\m} \|G_\lam0^\Lambda(\m, \n)\|^{1-\epsilon}  d\lambda
  \leq C'' g^{-1}\exp(-R\gamma/8) R^{d-1}~, \]
and the same estimate holds for the integral over $I_\n$. Therefore finally
\[\begin{split}
\langle Q_I(\m, \n) \rangle &\leq C^{IV} g^{-1} R^{d-1} \exp(-R\gamma/8) + C'k R^{2d}\exp(- \gamma s R/8 ) \\
  &\leq C^V R^{2d}\exp(- \gamma s R/8 )~.
\end{split}\]
\end{proof}


\begin{thebibliography}{99}

\bibitem{A} M.~Aizenman,
Localization at weak disorder: some elementary bounds,
Rev.\ Math.\ Phys.~6 (1994), no.~5A, 1163--1182.

\bibitem{AM} M.~Aizenman, S.~Molchanov,
Localization at large disorder and at extreme energies: an elementary derivation,
Comm.\ Math.\ Phys.~157 (1993), no.~2, 245--278.

\bibitem{AW} M.~Aizenman, S.~Warzel,
Random operators ({\em lecture notes in preparation})

\bibitem{B} J.~Bourgain,
An approach to Wegner's estimate using subharmonicity,
J.\ Stat.\ Phys.~134 (2009), no.~5-6, 969--978.

\bibitem{CHK}
J.-M.~Combes, P.~Hislop, and F.~Klopp,
An optimal {W}egner estimate and its application to the global continuity of the integrated density of states for random Schr\"odinger operators. Duke Math. J., 140 (2007), 469--498.

\bibitem{CFKS} H.~L.~Cycon, R.G.~Froese, W.~Kirsch, and B.~Simon,
Schr\"{o}dinger Operators,
Berlin, Heidelberg, New-York: Springer 1987

\bibitem{EKTH} A.~Elgart, H.~Kr\"uger, M.~Tautenhahn, I.~Veseli\'c,
Discrete Schr\"odinger operators with random alloy-type potential,
{\em to appear in} Proceedings of the Spectral Days 2010,
Pontificia Universidad Cat\'olica de Chile, Santiago,
{\em preprint:} \url{arxiv:1107.2800}.

\bibitem{ETV} A.~Elgart, M.~Tautenhahn, I.~Veseli\'c,
Anderson Localization for a Class of Models with a Sign-Indefinite
  Single-Site Potential via Fractional Moment Method,
Annales Henri Poincar\'e~12:8 (2010), 1571--1599

\bibitem{FS} J.~Fr\"ohlich, T.~Spencer,
Absence of diffusion in the Anderson tight binding model for large disorder or low energy,
Comm.\ Math.\ Phys.~88 (1983), no.~2, 151--184.

\bibitem{GM} M.~Gebert, P.~M\"uller,
Localization for random block operators, 
\url{http://arxiv.org/abs/1208.1947}.

\bibitem{GK} F.~Germinet, A.~Klein,
New characterizations of the region of complete localization for random Schr\"odinger operators, 
J.\ Stat.\ Phys.\ 122 (2006), no.~1, 73--94. 

\bibitem{GMP}
Ya.~Goldsheid, S.~Molchanov, and L.~Pastur,
Pure point spectrum of stochastic one dimensional Schr\"odinger operators. Funct. Anal. Appl. 11, (1977), 1--10.

\bibitem{HS} H.~V.~Henderson, S.~R.~Searle,
On deriving the inverse of a sum of matrices,
SIAM Rev.~23 (1981), no.~1, 53--60.

\bibitem{Kirsch} W. ~Kirsch,
An invitation to random Schr\"odinger operators, Random Schr\"odinger operators, Panor. Synth\`eses, vol. 25, Soc. Math. France, 2008, With an appendix by Fr\'ed\'eric Klopp, pp. 1--119.

\bibitem{KMM} W.~Kirsch,
B.~ Metzger, and P.~M\"uller,
Random block operators,
J. Stat. Phys.~143 (2011), pp. 1035--1054.

 \bibitem{Klopp}
 F.~Klopp,
Weak disorder localization and {L}ifshitz tails.
Comm. Math. Phys.~232 (2002), no.~1, 125--155.

\bibitem{Stollmann}
P.~Stollmann,
Caught by disorder. Bound states in random media. Progress in Mathematical Physics. 20. Boston: Birkh\"auser (2001).

\bibitem{Stolz2000}
G.~Stolz,
Non-monotonic random Schr\"odinger operators: The Anderson model.
J.\ Math.\ Anal.\ Appl.~248 (2000), no.~1,  173?-183.

\bibitem{Stolz2010}
G.~Stolz,
An introduction to the mathematics of {A}nderson localization.
Entropy and the quantum {II}, Contemp. Math., 552, pp. 71--108, Amer. Math. Soc., Providence, RI (2011).

\bibitem{W} F.~Wegner,
Bounds on the density of states in disordered systems,
Z.~Phys.~B~44 (1981), no.~1--2, 9--15.

\end{thebibliography}
\end{document}